\documentclass{article}

\usepackage{stmaryrd,amssymb,amsmath,amsthm,url,xfrac}
\usepackage{graphicx,version}
\usepackage[usenames]{color}

\providecommand{\dotdiv}{% Don't redefine it if available
  \mathbin{% We want a binary operation
    \vphantom{+}% The same height as a plus or minus
    \text{% Change size in sub/superscripts
      \mathsurround=0pt % To be on the safe side
      \ooalign{% Superimpose the two symbols
        \noalign{\kern-.35ex}% but the dot is raised a bit
        \hidewidth$\smash{\cdot}$\hidewidth\cr % Dot
        \noalign{\kern.35ex}% Backup for vertical alignment
        $-$\cr % Minus
      }%
    }%
  }%
}
\theoremstyle{definition}

\newtheorem{definition}{Definition}[section]

\newtheorem{proposition}{Proposition}[section]

\newtheorem{problem}{Problem}[section]
\newtheorem{example}{Example}[section]

\newcommand{\LLOC}{\mathsf{LLOC}}
\newcommand{\NOS}{\mathsf{NOS}}
\newcommand{\gsc}{\!\mbox{\Large;}}

\begin{document}

\title{Quantitative Expressiveness of Instruction Sequence Classes for Computation on Single Bit Registers}

\author{%
Jan A.\ Bergstra\thanks{Informatics Institute, University of Amsterdam, 
email: \texttt{j.a.bergstra@uva.nl, janaldertb@gmail.com}.}
}

\date{April 18, 2019}

\maketitle

\begin{abstract} The number of instructions of an instruction sequence is taken for its logical SLOC, 
and is abbreviated with LLOC. A notion of quantitative expressiveness is based on LLOC and in the special case of 
operation over a family of single bit registers a collection of 
elementary properties are established. A dedicated notion of interface is developed and is used for stating 
relevant properties of classes of instruction sequences.\\

\noindent ACM classes: F.1.1; F.2.1.

\end{abstract}

\tableofcontents

\section{Introduction}
This paper makes use of the theory and notation regarding instruction sequences for operation on 
Boolean registers as surveyed in~\cite{BergstraM2018} for the special case of operations on Boolean registers thereby 
following the notation of~\cite{BergstraM2016b} and simplifying the general presentation of~\cite{BergstraM2012} and~\cite{BergstraP2007b}. 

Existing notations and results regarding instruction sequences will be used mostly without further reference or technical introduction because such expositions having amply been published. 
We mention~\cite{BergstraM2007,BergstraM2012,BergstraM2012b,BergstraM2012d,BergstraM2014a,BergstraM2016b}
and \cite{BergstraM2018}), and further references listed in these papers.
For the following notions, terms and phrases, we refer to the  papers just mentioned and the references contained in those: 
basic instruction ($a \in A$), focus, method, focus method notation for basic instructions 
($a = f.m$ with focus $f$ and method $m$), yield  (also called reply) of a basic instruction(a), positive test instruction ($+a$), negative test instruction ($-a$), termination instruction ($!$), (forward) jump instruction ($\#k$), backward jump instruction ($\backslash\#k$), in direct jump instruction, finite PGA instruction sequence, (alternatively: single pass instruction sequence or PGA instruction sequence without iteration),  PGLB program (PGA instruction sequence with with backward jumps 
instead of iteration), generalised semi-colon (text sequential composition), 
thread, terminated thread (stopped thread $S$), diverging thread ($D$),  
thread extraction from an instruction sequence ($|X|$ for an instruction sequence $X$), service, service family, 
empty service family, service family composition operator, service family algebra, apply operator ($-\bullet-$),
 the method interface $M_{16}$ consisting of
16 methods of the form $y/e$ for Boolean registers with $y,e \in \{0,1,i,c\}$ ($y$ for yield, $e$ for effect). 

\subsection{Logical lines of code for an instruction sequence}
Because the identification of Booleans and bits may lead to confusion Boolean registers will be referred to as 
single bit registers below. 

The number of instructions of an instruction sequence is referred to 
as its length in e.g.~\cite{BergstraM2014a,BergstraM2016}.
However, in order to develop a terminology which is
more similar to the classical notion of LOC (lines of code, also referred to as SLOC for source lines of code) 
we will make use of the following terminology:
\begin{definition} LLOC (logical lines of code): for an instruction sequence $X$, written 
in any PGA-style instruction sequence notation, 
$\LLOC(X)$ 
\label{lloc} denotes the number of instructions of $X$.
\end{definition}
Conventions for the notation of instructions are such that $\LLOC(X)$ equals the number of semi-colons in $X$ plus one.
\begin{definition} An instruction sequence has low register indices if for each kind of register the collection of register 
numbers of registers involved in one or more of its basic instructions  constitute an initial segment of the positive natural numbers.
\end{definition}

$\LLOC(X)$ is not a precise measure of the size of $X$ in terms of bytes. A reasonable estimate is that for the instruction sequence notations used below, and assuming that the instruction sequence has low register indices, the size of $X$ as measured in bytes will not exceed say $200 \cdot \LLOC(X) \cdot\, ^{10}\log \LLOC(X)$.

We refer to~\cite{NguyenDTB2007} for an 
exposition on various forms of LOC and SLOC in software engineering practice. 
In the setting of PGA style instruction sequences
no distinction between a statement and an instruction is made and LLOC according to 
Definition~\ref{lloc} is a plausible interpretation of logical SLOC which 
is characterised in~\cite{NguyenDTB2007} as a metric, or rather a family of metrics, 
based on counting the number of statements in a source code.
LLOC as in Defnition~\ref{lloc} comes close to the metric used implicitly in~\cite{Massalin1987}.

\subsection{Existing approaches to program size}
Work on program size has been carried out in the setting of computability theory, for 
instance~\cite{Constable1971},~\cite{Meyer1972}, and~\cite{KatseffS1981} in relation to Kolmogorov complexity. 
In~\cite{Mehlhorn1982} program size is defined as the set of characters of a program and it is 
related with practical computational tasks, while~\cite{Chaitin2003}
links program size with information theory. Unlike these approaches we use a rather fixed family of 
program notations, viewing a program as a sequence of instructions. 
By taking the number of instructions as a metric full precision is obtained while at the same time 
abstraction from the ad hoc syntax of instructions is achieved.

\subsection{Objectives of the paper}
\label{objectives}
The objective of this paper is to describe some elementary quantitative observations pertaining 
to instruction sequences and LLOC metric under the simplifying assumption that basic actions operate on a family of 
single bit  registers, which arguably  are the simplest conceivable datastructures. 
We will assume that the semantics of an instruction sequence, i.e. what it computes,
  is a partial function from tuples of bits to tuples of bits, thereby excluding instruction sequences meant for 
  computing interactive systems.

We will demonstrate that for very simple tasks determination of the lowest LLOC of an implementation for that task is possible,
and we will show by means of examples that theoretical work on LLOC minimisation is greatly facilitated by being explicit 
about the precise method interfaces of various  single bit registers. 
For each bit there are $2^{16}$ possible interfaces and
therefore when designing an instruction sequence for a task $F$ involving $n$ input registers, and $m$ output registers,
while allowing the use of an arbitrary number of auxiliary registers, each with the same method interface, a total of 
$2^{16\cdot n\cdot m}$
different combinations of method interfaces each constitute potentially different versions of the problem to implement $F$
and to do so with a minimal (or relatively small) LLOC count. Many questions are stated and left unanswered.

A single bit register is a service (program algebra terminology for a system component able to execute the actions
of an instruction sequence)  which is accessed by a calling instruction via its focus. 
A focus plays the role of the name of a service, and at the same time it is informative about the role of the service. 
Below we will mainly consider the following foci: $in{:}i$ and $out0{:}i$ for $i \in \mathbb{N}$. 
The inputs for a computation are placed in the registers $in{:}i$ (so-called input registers), 
the outputs of a computation are found in the registers  $out0{:}i$ ($0$ initialised output registers). 
At the end of a computation the final value of the input registers is forgotten.
The focus prefixes $in$ and $out0$, are referred to as register roles. Other register roles exist, for instance 
 $out1$ for output registers which have have initial value $1$, $inout$ for a register which serves both as an input and as an output,
$aux0$ for an auxiliary register with initial value $0$, and $aux1$ for an auxiliary register with initial value $1$.
\subsection{Quantitative expressiveness versus qualitative expressiveness}
Expressiveness of a formalism for denoting instruction sequences may be measured in many ways. 
We will mainly consider the following idea: given a task, computing a total or partial function 
$F$ of type $ \mathbb{B}^n \to  \mathbb{B}^m$ we are interested in the shortest instruction sequence(s), 
taken from some class $K$ of instruction sequences, 
that is instruction sequences with a minimal number of instructions, which compute $F$. 
Clearly if $K_1 \subseteq K_2$ are two classes of instruction sequences then $K_2$ 
may be considered more expressive (more expressive w.r.t. LLOC) than $K_1$ if for some task $F$
all instruction sequences in $K_1$ that compute $F$ are longer than $n$ 
with $n$ the minimal LLOC for an instruction sequence in $K_2$ that implements task $F$.

\begin{definition} Let $K_1 \subseteq K_2$ be classes of instruction sequences. 
$K_2$ is more expressive (more expressive w.r.t. LLOC) than $K_1$ if for some task $F$
there is an instruction sequence $X \in K_2$ which computes $F$ while there is no instruction sequence 
$Y \in K_1$ which also computes $F$ such that $\LLOC(Y) \leq \LLOC(X)$.
\end{definition}
In some cases the smaller class of instruction sequences does not provide any implementation for a task 
which is implementable with the larger class. Then I will speak of differentiation of qualitative expressiveness.
\begin{definition} Let $K_1 \subseteq K_2$ be classes of instruction sequences. 
$K_2$ is qualitatively more expressive  than $K_1$ if for some task $F$
there is an instruction sequence $X \in K_2$ which computes $F$ while there is no instruction sequence 
$Y \in K_1$ which also computes $F$.
\end{definition}

\subsection{Rationale of designing additional forms of instructions}
Below several types of instructions outside the core syntax of PGA will be discussed: 
instructions for structured programming, 
backward jumps, indirect jumps, and generalised semi-colon instructions. 
These constitute merely a fraction of the options for extension of the syntax of instruction
sequences that have been explored in recent years.

We will assume that  that the rationale of the introduction of 
additional kinds of instructions is to achieve one or more of of four potential advantages, 
upon making use of the ``new'' instructions:
\begin{description}
\item [Fewer instructions.]
Some tasks may be implemented with a shorter instruction sequence, that is with fewer instructions. 
(This criterion when applied in practice amounts to the optimisation of program size or achieving good code compactness.)
\item [Fewer steps.]
A given task may be implemented by an instruction sequence which produces faster runs, i.e. fewer steps are taken till termination, either in the worst case or in average or according to some other efficiency criterion.
\item [Fewer mistakes.]
Correct or `high quality'' instruction sequences can be produced either more quickly, or in a more readable form or, in 
such a manner that some given form of analysis or verification is more easily applied, 
or can be applied with a higher rate of success.
\item [Fewer compiler optimisations.]
A given task may be implemented by an instruction sequence which allows the production of efficient compiled version with
fewer optimisation steps.
\end{description}
Below we will focus exclusively on the first two advantages. Undoubtedly the third advantage may become harder to achieve when optimising either code compactness or execution speed or both.

\subsection{Generalised semi-colon and a non-expanding LLOC metric}
We will make use of generalised semi-colon notation: $\gsc_{k=1}^n(X_k)=X_1;\ldots;X_n$. In order to apply
the  LLOC metric the generalised semi-colons must be expanded first. 

An alternative presentational metric $\LLOC_{gsc}$, called the generalised semi-colon non-expanding LLOC metric, 
works as follows: (i) for $X$ not containing any occurrence of the generalised sequential composition construct: 
$ \LLOC_{gsc}(X)= \LLOC(X)$, $ \LLOC_{gsc}(X:Y)= \LLOC_{gsc}(X)+ \LLOC_{gsc}(Y)$, and (iii) 
$\LLOC_{gsc}(\,\gsc_{k=1}^{n} (Y_k))=  \LLOC(Y_k)+ 2 + \lfloor ^2 \log n \rfloor $.
The idea is to count ``$\,\gsc_{k=1}^{n} ($'' as well as the corresponding closing bracket ``$)$'' 
as if these were instructions, and to add a logarithmic increment taking account for the size of $n$.

When writing an instruction sequence the use of generalised semicolon notation may improve readability. It may also be 
easier to write a compiler for instruction sequence expressions involving generalised semi-colons than for 
expanded versions thereof.

\subsection{Terminology and notation for roles}
\label{terminology}
The strings $in$, $out$, $inout$, and $aux$ serve as role headers which prefix the role base, 
whereas $0$ and $1$ are role postfixes which may be appended
to the role base. A role comes about given a role base from prefixing the role base with the header and in addition,
in case the header is either $out$or $aux$, postfixing the result with a postfix. 

For single bit registers the preferred role base is the empty string and the respective roles are 
$in,inout,out0,out1,aux0,aux1$. 
Corresponding foci include a further number so that different copies of services for the role at hand 
can be distinguished. Examples of foci for the various roles for 
single bit registers are are e.g. $in{:}7, inout{:}13, out0{:}3$.
Below we will introduce instructions with role base $1D$ for 1 dimensional single bit arrays, and we will use additional 
role bases $\_ a$ and $\_ b$ (with foci e.g. $in\_a{:}2m, in\_b{:}5, inout\_a{:}1$)  in order to enhance readability.
\section{Expressiveness of single pass instruction sequences}
All functions from bit vectors of length $n$  to bit vectors of length $m$ 
can be computed without the use of backward jumps, that is 
 without the use of any form of iteration or looping. Proposition~\ref{Expressiveness1} was shown for $m=1$ in
 taken from~\cite{BergstraM2014a}, the extension to $m>1$ is straightforward.
  The following function $l$  will be used. 
\begin{definition}  $l \colon \mathbb{N} \times \mathbb{N} \to \mathbb{N}$ is given by 
$l(0,m) = m+1$, $l(n+1,m) = 2 \cdot l(n,m)+2$.
\end{definition}
\begin{proposition} 
$l(n,m) = 2^n \cdot (m +3) -2$.
\end{proposition}
\begin{proof} Induction on $n$. The case $n=0$ is immediate. Step: $l(n+1,m) = 2 \cdot l(n,m) + 2= 
2 \cdot (2^n \cdot (m +3) -2) +2= (2^{n+1} \cdot (m +3) -4) +2= 2^{n+1} \cdot (m +3) -2$.
\end{proof}

\begin{proposition} 
\label{Expressiveness1} 
Let $n, m \in \mathbb{N}$, with $m >0$.  For each total $F \colon  \mathbb{B}^n \to  \mathbb{B}^m$ 
there is a finite PGA instruction
sequence  (i.e. single pass instruction sequence, or instruction sequence without iteration) $X^F$ 
with basic instructions 
of the form $f.y/e$ with focus $f \in \{in{:}1,\ldots,in{:}n,out0{:}1,\ldots,out0{:}m\}$ and 
method $y/e \in M_{16}$ which computes $F$. Moreover the $X^F$'s can be chosen such that
$\LLOC(X^F) = l(n,m)$.
\end{proposition}

\begin{proof} We will use induction on $n$. If $n=0$, 
$F$ produces a sequence of constants $(\mathrm{d}_1,\dots,\mathrm{d}_m)$
which is computed by  $X^F = \gsc_{k=1}^m (out0{:}1.1/\mathrm{d}_k);!$. We notice that $\LLOC(X^F) = m+1 = l(0,m)$. 

Now consider the case $n+1$. We split $F$ into $F_0$ and $F_1$ such that for all $\vec{b}$, $F(\vec{b},0) = F_0(\vec{b})$ and $F(\vec{b},1) = F_1(\vec{b})$. Using the induction hypothesis one may finc 
$X_0$ and $X_1$ with $l(n,m)$ instructions  each and 
which compute $F_0$ and $F_1$ respectively. Now the instruction sequence 
 $X = +in{:}(n+1).i/i;\#(l(n)+1);X_0;X_1$ computes $F$and that it $\LLOC(X)=l(n+1,m)$.  
\end{proof}
The design of $X$ has many alternatives. For instance setting
$X = +in{:}(n+1).i/0;\#(l(n)+1);X_0;X_1$ works as well, 
while its basic actions are less amenable to reading input from the  input registers 
as the value of input registers is set to 0 at the first method call to it.
 
\subsection{Computational metrics: NOS}
With $\NOS(X,H)$ (number of steps) I will indicate the number of instructions that is processed during the (unique) 
run of the instruction sequence $X$ on service family $H$. 
If divergence occurs, a jump with counter 0, or a jump outside the range of instructions, $\NOS$ takes the value $\infty$.
Equally if an error occurs, i.e. the occurrence of 
a method call outside the interface  provided by $H$, $\NOS$ takes the value $\infty$. 
Interfaces are discussed in detail in Section~\ref{Interfaces}.

Some examples of $\NOS$: $\NOS(!,H) = 1, \NOS(\#1;\#1;!,H) = 3, \NOS(\#1;\#0;!;!,H) = \infty$\\
$NOS(+out0{:}1.1/1;!, out0{:}1.br(1))=2$, $NOS(+out0{:}1.1/1;\backslash \# 2;!, out0{:}1.br(1))=2$\\
$ \NOS(\#2;!,H) = \infty$, and $ \NOS(+in{:}3.i/i ;!, out0{:}1.br(0)) =\infty$.

The instruction sequence $X$ which has been constructed in the proof of Proposition~\ref{Expressiveness1}  computes $F$ in such a manner that each result is found in precisely  $2^\cdot (m+3) -2$ 
steps. This is average, and worst case $\NOS$ of  $2^\cdot (m+3) -2$ for $n$ inputs and $m$ outputs. This figure for $\NOS$
 is fairly  low as, except from producing outputs, it provides  just one instruction on average to process 
each bit of input after it has been read. 

\subsection{Tradeoff between LLOC and NOS: an open question}
The implication of this observation is that it is easy to write an instruction sequence $X$ which produces fast computations, 
i.e. a low worst case $\NOS(X,H)$ fro relevant $H$ while it may be hard to ensure 
in addition that $\LLOC(X)$ is kept reasonably small. If $\LLOC(X)$ entails a combinatorial explosion, 
then so does the activity of designing and constructing $X$. 

In other words: given a task $F$ (with $n$ inputs and $m$ outputs, both fixed numbers) the 
programming problem to write an instruction sequence implementing the task primarily constitutes a 
challenge to find an implementation with low LLOC, a state of affairs which brings LLOC to prominence. 
It is unclear to what extent minimising LLOC stands in the way of obtaining a low worst case or average NOS
in practice, that is for meaning full tasks $F$. The following question, for which we have no answer, 
constitutes one of many ways to formalise this matter.

\begin{problem} Is there a family of functions 
$F^n \colon \mathsf{B}^n \to \mathsf{B}$ for which LLOC 
minimal implementing instruction sequences (admitting auxiliary registers) 
have superpolynomial worst case NOS performance?
\end{problem}

However, as minimising $\LLOC$ is in most cases an unfeasible challenge, it is reasonable to look for a combined metric. 
We are unaware of a plausible candidate for a combined metric, however, which leads us to stating the following 
conceptual question.

\begin{problem} Find a plausible metric for instruction sequences (which measures the success of a design) and 
combines $\LLOC$ and $\NOS$ by capturing a useful tradeoff between these. 
\end{problem}

In~\cite{BergstraM2018} we have presented various designs of single pass instruction sequences for multiplication of 
natural numbers in binary notation. 
As it stands we have no systematic method to assess the success of these designs in quantitative terms. 
The processing speed (low worst case NOS) which is achieved by way of a divide and conquer 
approach is relevant only if the cost in terms of LLOC is not too high, 
and we have no obvious way to assess that matter.

A way out of this matter is to insist that implementing a family of functionalities $F^{n,m}\colon \mathsf{B}^n \to \mathsf{B}^m$
by means of single pass instruction sequences $X_{n,m}$ must be done under the additional requirement that
$\LLOC(X_{n,m}) \in O((n+m)\cdot \log n \cdot \log m)$. That requirement, however, 
rules out the instruction sequences found in~\cite{BergstraM2018}.
\subsection{Backward jumps and LLOC, an open problem}
One may incorporate iteration by allowing backward jump instructions (written $\backslash \#k$). 
PGLB is the instruction sequence notation which admits the instructions from PGLA (without iteration) 
as well as backward jump instructions.
Thus PGLB instructions are $!,\#k,\backslash\#k, a,+a,-a$ with $a$ a basic action, 
i.e. an action of the form $f.m$ with $f$ a focus and $m$ a method.

\begin{proposition} There is a computable translation $\Psi$ which transforms PGLB instruction sequences for 
single bit  registers into finite PGA instruction sequences working on the same single bit  registers in such a manner that
for each $X$, $\Psi(X)$ computes the same function as $X$ on the single bit  registers which it makes use of.
\end{proposition}
\begin{proof} Given PGLB instruction sequence $X$ working on $n$ inputs, all $2^n$ input vectors are presented to $X$ 
and the results are computed and collected in an appropriate finite datastructure. Now the proof of 
Proposition~\ref{Expressiveness1} is understood as the description of an algorithm by means of which the required 
instruction sequence is
created.
\end{proof}

\begin{proposition} If there exists a  translation $\Psi$ which transforms each PGLB instruction sequence  for single bit  
registers with low register indices into a finite PGA instruction sequence with low register indices working on the same single bit  
registers, perhaps making use of additional auxiliary registers, in such a manner that 
(i) for each $X$, $\Psi(X)$ computes the same function as $X$, and (ii) for each $X$, 
$\LLOC(\Psi(X))$ is bounded by a fixed polynomial $p(-)$ in LLOC$(X)$, 
then NP $\in$ P/Poly, and in fact NP $\in$ $P/O(n \log n)$
\end{proposition}
\begin{proof}
The connection between instruction sequences and complexity theory with advice functions has been 
explored in detail in~\cite{BergstraM2014a}. 
The idea is that one may understand the instruction sequence itself as an advice function. 
The proof is an elementary application of the results in the mentioned paper. 
As a bit sequence the instruction sequence is of polynomial length in its number of instructions. The mechanism to 
compute the result of the execution of a single pass instruction sequence with low register indices on given inputs is 
$O(n \log n)$ with $n$ the $\LLOC$ of $X$.
\end{proof}

It follows from these observations that it is implausible that for each PGLB instruction sequence a finite PGA
instruction sequence of equal LLOC size can be found which computes the same function. 
%Indeed, if it can be found it can be found in a computable manner, the problem lies in existence.

Upon taking into account the presence of more powerful services it is easily possible to demonstrate that backward jumps 
allow to write shorter programs for certain tasks. 
This idea is pursued in detail for certain services that represent an array of bits, i.e. indirect addressing of single bit registers.

\subsection{The simplest array: using a single bit as an index}
The simplest Boolean array has two single bit registers. 
Its role base is $1Da$ (for 1 dimensional array), role headers are $in, inout,out,aux$. 
The method interface is as follows: (i) methods
$a1{:}y/m$ for $y/m \in M_{16}$ apply $y/m$ to the index bit $a1$, and (ii) methods $y/m$ for $y/m \in M_{16}$
can be used and will apply to the register indexed by the 
current value of $a1$. 

The initial value of both registers, when of relevance,  is given by the role postfix. The access bit 
is in fact also a part of this service kernel which has 8 states for that reason, 

For instance $out1D0{:}3$ is the focus for the $3$th output service of this kind for which it is 
required that the registers are initially set to $1$. (By consequence $out1D1{:}3$ is a different focus.)

Copying say $in1D{:}3$ to say $out1D0{:}7$ can be done as follows with $\LLOC=6$:\\
\noindent $Copy1D= +in1D{:}3.i/i;out1D0{:}7.1/1;out1D{:}7.1/c;+in1D{:}3.a1{:}i/c;!;\backslash 5$.

A lower bound on $\LLOC$ for array copying in dimension 1, for a single pass instruction sequence, however, is  
$7$: an access method for each of  both arrays  must appear twice (4), 
a method application  the  index bits of both arrays is necessary (2) and 
at least one termination instruction (1) is required. We find:
\begin{proposition} In the presence of 1D single bit addressed single bit arrays 
the use of backward jumps increases the expressive power of the instruction sequence notation.
\end{proposition}

The following question is open: 

\begin{problem} Is it the case that the introduction of backward jumps in addition to 
PGA instructions renders the instruction sequence notation more expressive in terms of 
allowing to compute some functions with smaller LLOC size, 
and for the purpose of computing Boolean functions.
\end{problem}
The stated question is not very specific for the precise syntax that allows for repetition. 
As an alternative to backward relative jumps one might consider: 
absolute jumps (see~\cite{BergstraL2002}), goto's with label instructions 
(see~\cite{BergstraL2002}) and indirect jumps (see~\cite{BergstraM2007}).

We expect that multiplication of two $n$-bits natural numbers thereby producing $2n$ output values 
constitutes a task for which the
availability of backward jumps provides a provable advantage in terms of the minimisation of the LLOC metric. 
This phenomenon may well appear for fairly low $n$, say $n=5$ or below. 

\subsection{Unfolding an instruction sequence with backward jumps}
Given an instruction sequence $X$ in PGLB notation, i.e. with backward jumps, 
one obtains $Y$ which computes the same function as follows. Let $X=u_1;\ldots;u_n$. 
In $X$ obtain $Z$ by replacing each jump $u_i$ by $u^\prime_i$ as follows:
\begin{itemize}
\item if $u_i \equiv \#k$ and $i+k\leq n$ then $u^\prime_i \equiv u_i$,
\item if $u_i \equiv \#k$ and $i+k > n$ then $u^\prime_i \equiv \#0$,
\item if $u_i \equiv \backslash \#k$ and $k \geq i$ then $u^\prime_i \equiv \#0$,
\item if $u_i \equiv \backslash \#k$ and $k < i$ then $u^\prime_i \equiv \#(n-k)$.
\end{itemize}
Then take $Y \equiv Z^\omega$. Assuming that $X$ works on a finite domain, 
for some $p>0$, $Z^p$ ($Z;\ldots;Z$, $p$ 
consecutive copies of $Z$) is a finite PGA instruction sequences which 
computes the same function as $X$. Moreover the computations take 
precisely as many steps as for $X$. In Section~\ref{Cmetrics} below we will discuss 
a computational metric for which $X$ and $Z^p$ are equivalent, assuming that $p$ is taken sufficiently large. We notice: 
$\LLOC(Z^p) = p\cdot \LLOC(X)$.
irect address options for the array. Moreover a termination instruction is needed. Together this makes up for 
al $CopyA_k=$ has $4.k +2$ instructions. For $k=3$ this yields $14$ which is well below  the $2^5-1 = 31$
 (for array copying in dimension 2) in the presence of a backward jump is easily found to be 

\section{Proper subclasses of single pass instruction sequences}
In this section we consider two restrictions on the design of instruction sequences in relation to expressiveness.
The first restriction is that no register is acted upon more than once. The second restriction imposes an upper
bound on te size of jumps. In Paragraph~\ref{simplifying} below we will 
consider a third proper subclass of instruction sequences, by disallowing intermediate termination.

\subsection{Single visit single pass instruction sequences}
A useful subclass of finite PGA instruction sequences consists of those instruction sequence which contain 
at most one method call for each register. We will refer to these instruction sequences as single visit instruction sequences. 

The single visit restriction comes with consequences for qualitative expressiveness.
Consider the instruction sequence $X$ with two inputs and one output. \\
$X = +in{:}1.i/i;\#3;-in{:}2.i/i;\#2;+in{:}2.i/i;out0{:}1.1/1;!$.\\
$X$ computes the function $ F(in{:}1, in{:}2) = in{:}2 \lhd in{:}1 \rhd (\neg\, in{:}2)$. As it turns out
imposing the requirement that single pass instruction sequences are also single visit instructions reduces the
qualitative expressiveness of the system.

\begin{proposition}
The  function $ F(in{:}1, in{:}2)$ as mentioned above cannot be computed by a single visit single pass PGA
 instruction sequence.
\end{proposition}

\begin{proof} For single visit instruction sequences the use of auxiliary registers is not relevant as the first and last method
call to it, if any call is made, will only return the know initial value of an auxiliary register.
Assume that (i) $Y$ is a single visit single pass PGA instruction sequences which has the required functionality,
(ii) $Y$ contains at most one call to each of the three single bit  registers involved, and (iii) the first method call to a register 
in $Y$ is for $in{:}1$. Now notice that after reading $in{:}1$ for both replies the intended output still depends 
on the content of $in{:}2$. Thus in both cases at some stage (i.e. after $0$ or more jumps) 
some test instruction takes input from $in{:}2$. As there is only a single test instruction for 
$in{:}2$ in $Y$ it follows that irrespective of the outcome of
the test in $in{:}1$ that test on $in{:}2$ is performed. As a consequence the result of the computation 
of $Y$ cannot depend on the initial content of $in{:}1$ which is wrong. So one may assume that the first 
call is to register $in{:}2$. Now the output still depends in the value of $in{:}1$ and therefore in both cases the unique call to
$in{:}1$ is reached and the output must be determined after processing that instruction so that it will not depend on the value
read from $in{:}2$. \end{proof}
 
\subsection{Single pass instruction sequences with bounded jumps}
Another plausible restriction on single pass PGA instruction sequences results from imposing an 
upper bound to the size of jumps. At the moment of writing we have no answer concerning the following question.

\begin{proposition}
\label{TwoJumps}
Each Boolean function with finite range and domain be computed by a single pass PGA instruction 
sequence that involves jumps of size at most 2.
\end{proposition}

\begin{proof} We consider a function $F(-,-,-)$ taking its arguments from registers $in{:}1, in{:}2, in{:}3$
and producing results $F_1(-,-,-), F_2(-,-,-)$ and $F_3(-,-,-)$ in registers $out0{:}1, out0{:}2$ and $out0{:}3$. 
The construction is done in such a manner that it generalises to all cases.

Let $\alpha^1,\ldots,\alpha^{2^3}$ be an enumeration of the arguments of $F$. We write 
$\alpha^i = (\alpha^i_1,\alpha^i_2,\alpha^i_3)$. An instruction sequence $X_F$ computing $F$ is found as follows:
$$X_F= \mbox{\Large ;}_{k=1}^{2^3}(test_F^{\alpha^k};+out0{:}1.0/F_1(\alpha);\#2;+out0{:}2.F_2(\alpha);\#2;+out0{:}3.F_3(\alpha));!$$ with $test_F^{\alpha} =-in{:}1.(i\lhd \alpha_1\rhd c)/i;\#2; -in{:}2.(i\lhd \alpha_2\rhd c)/i;
\#2;-in{:}3.(i\lhd \alpha_3\rhd c)/i;\#2$.
\end{proof}

Following~\cite{BergstraM2012} an instruction sequence with jumps of size $1$ only can be transformed into an equivalent instruction sequence without jumps. Thus the use of jumps of size $1$ does not increase expressiveness. Moreover, 
In the presence of auxiliary registers can be avoided.

\begin{proposition}
\label{nojumps}
With the use of arbitrarily many auxiliary registers (say $aux0{:}1,\ldots,aux0{:}k$) each function on single bit  registers can be computed by a single pass PGA instruction sequence without jumps. 
\end{proposition}
\begin{proof} Using Proposition~\ref{Expressiveness1}, given $F$,  some single pass PGA instruction 
sequence $X$ over registers used by $F$ may be chosen such that $X$ computes 
$F$. Using  the main result of~\cite{BergstraM2012}, with the help of sufficiently may auxiliary registers 
$aux0{:}1,\ldots,aux0{:}k$ a single pass instruction sequence $Y$ is found such that 
$X / \bigoplus_{l = 1}^{k} aux0{:}l.br(0) = |Y|$ (after abstraction from internal steps). It follows that 
$|X| \bullet H \oplus  \bigoplus_{l = 1}^{k} aux0{:}l.br(0) = |X| \bullet H$ which implies that that $X$ computes $F$.
\end{proof}
Although large jumps are not required for the computing any Boolean function, it still 
may be the case that imposing a restriction to small jumps leads to the need for 
longer instruction sequences, or it may imply the need for the use of more auxiliary registers. 
\begin{example}
\label{Exam1}
Consider the function $G^k$ with one input $in{:}1$ and $2k$ outputs 
$out0{:}1,\ldots,out0{:}k$,\\ $out0{:}k+1,\ldots,out{:}k+k$. $G^k(0)$ returns with 
$out0{:}1,\ldots,out0{:}k$ each set to $1$ while 
$G^k(1)$ returns with $out0{:}k+1,\ldots,out0{:}k+k$ set to $1$. $G^k$ is computed by:
$$X^k_G = -in{:}1.i/i;\#k+2;\gsc_{l=1}^{k} (out0{:}l.1/1);!;\gsc_{l=1}^{k} (out0{:}(k+l).1/1);!$$
$\LLOC(X^k_G) = 2k+3$ but in this case the jump instruction can be avoided, thereby achieving $\LLOC= 2k+2$, as follows:
$$Y^k_G = +in{:}1.i/i;\!\mbox{\Large ;}_{l=1}^{k-1} (-out0{:}l.1/1;-out0{:}(k+l).1/1);-out0{:}k.1/1;out0{:}(2k).1/1;!$$
None of the instructions of $Y^k_G$ can be avoided for any instruction sequence able to compute $X^k_G$.
It follows that
$Y^k_G$ is demonstrably a shortest instruction sequence able to compute $G^k$.
\end{example}

\begin{example}
\label{Exam2}
Now the example is modified by having additional inputs which govern whether or not the outputs are to be set to 1. Moreover
these additional inputs serve also as outputs and are complemented with each call,
For focus naming non-empty role bases (see Paragraph~\ref{terminology}) $\_a$ and $\_b$ are used and the function $E^k$ is 
computed by $X_E^k$ with $\LLOC(X_E^k)=4k+4$:

\noindent
$$X_E^k =+inout{:}1.i/c;\#2k+2;\gsc_{l=1}^{k} (+inout\_a{:}l.i/c;out0\_a{:}l.1/1);!;$$
$$\gsc_{l=1}^{k} (+inout\_b{:}l.i/c;out0\_b{:}l.1/1);! \texttt{~~~~~~~~~~~~~~}$$

\end{example}
It is easy to see that for $k \geq 1$ no instruction sequence with fewer than $4k+2$ instructions can compute $E^k$. From
Proposition~\ref{nojumps} we know that $E^k$ can be computed by an instruction sequence without jumps 
with the use of auxiliary registers, and from Proposition~\ref{TwoJumps} we know that it can be computed by means of
an instruction sequence involving jumps with length $2$ or less. The latter instruction sequence may be quite long, however. 
By admitting jumps of size $3$ LLOC $5k + 4$ can be achieved:\\

\noindent
$$Y_E^k =+inout{:}1.i/c;\gsc_{l=1}^{k} (\#3;+inout\_a{:}l.i/c;-out0\_a{:}l.1/1);\#2;!;$$
$$\gsc_{l=1}^{k} (+inout\_b{:}l.i/c;out0\_b{:}l.1/1);!\texttt{~~~~~~~~~~~~~~~}$$
It is plausible that for increasing $k$ the shortest single pass 
PGA instruction sequences for computing $E^k$ must involve 
increasingly large jumps, as does $X^k_E$. 
Proving that to be the case is another matter, however. We will provide a partial result on that matter in 
Proposition~\ref{Exam2ctd} below, making use of 
interfaces in order to restrict the scope of the assertion and thereby to allow for its proof.

\section{Interfaces}
The setting of instruction sequences acting on service presents an incentive for the introduction and application 
of various forms of interfaces. 
Interfaces may may be classified and qualified in different ways. To begin with we distingiosh required interfaces and provided interfaces. A thread or an instruction sequence
comes with a required interface, whereas services and service families come with a provided interface.\footnote{% 
Mathematically speaking required interfaces and provided interfaces are the same, 
though when working with interface groups required interfaces and provided interfaces 
may be thought of as inverses w.r.t. composition.}

If a component is placed in a context, it is plausible to assume that the component comes with a required interface and that the context has a provided interface and to 
require for a good fit that the component's required interface is a subinterface of the context's provided interface.\footnote{%
The roles of component and context are not set in stone: if an instruction sequence $X$ computes over a service family $H$, the thread $|X|$ is placed in a context made up of $H$ (denoted $|X|\bullet H$), whereas if the instruction sequence $X$ uses the service
family $H$ by way of the use operator $-/-$ (denoted $|X|/H$) it is less plausible to take this view as the use operator is not
based on the assumption that $H$ provides a way of processing for each request (method call) that is required 
(issued) by $|X|$.} 
 
\subsection{Service kernels and  method interfaces}
\label{SKandMI}
A method interface is a finite set of methods (i.e. method names).
A service kernel $U$ is a state dependent partial function from methods to Booleans. The domain of that function 
is called the method interface of the service kernel, and it is denoted with $I_m(U)$. The method interface of $U$ is 
 supposed to be independent of the state of $U$. Applying method $m \in I_m(U)$ to $U$ produces a yield 
 $y(m,U) \in \mathbb{B}$ and an effect $e(m,U)$ which technically is another service kernel with $I_m(e(m,U)) = I_m(U)$. 
 It is a useful convention to write $U = U(s)$ thereby making a state $s$ explicit. 

Using notation with an explicit state we may write $y(m,U(s)) =y_m(s)$ and  $e(m,U(s) ) = U(e_m(s))$.  A service
kernel with an empty method interface is called degenerate or inactive. 
We will equate all inactive kernels denoting these with the constant $0$ which satisfies: $I_m(0) = \emptyset$.

\subsection{Method interface of a single bit register service kernel} 
There are 16 methods for single bit  registers, written $y/e$ (yield / effect), 
with $y,e \in \{0,1,i,c\}$. These are the methods applicable to any single bit  register kernel $br(b)$ 
(i.e. $br(-)$ with content $b \in \mathbb{B}$). We write $M_{16}$ for the collection of these. 

Each subset $J_m$ of $M_{16}$ constitutes a method interface. 
Consequently there are $2^{16}$ method interfaces for single bit  registers.

For $J_m \subseteq M_{16}$, $\partial_{J_m}(br(b))$ denotes the service kernel which admits 
precisely the methods of $M_{16} -I_m$ on
the single bit  register $br(b)$. It follows that 
for $b \in \{0,1\}$, $br(b) = \partial_{\emptyset}(br(b))$, so that 
$I_m(br(b)) = M_{16}$, and $I_m(\partial_{J_m}(br(b)))=J_m$.

\subsection{Focus kernel linking and service family composition}
A service kernel $U$ may be linked to (or: prefixed by, or: positioned under, or: combined with) a focus $f$ whereby a 
new service $f.U$ is obtained. If $f \neq g$ then $g.U$ is a different service starting out as a copy of $U$.

Service families are combinations of services created from the empty service family and services $f.U$ by way of 
service family composition (denoted $-\oplus-$) which is commutative and associative, 
and for which the empty service is a unit element.  Service family composition is not idempotent, 
however, as $f.U\oplus f.V = f.0$ (with $0$ as in Paragraph~\ref{SKandMI} above), 
thereby indicating that ambiguity in the
service provided by a context is considered problematic, rather than  it is resolved in a non-deterministic manner.
Indeed if services with the same focus are
combined an ambiguity arises as to which service kernel is to process $m$, 
and for this dilemma no simple solution exists. For that reason the combination
$g.U \oplus g.W$ is understood as an error in the algebra of service families.
 
When combining services $f.U$ and $g.W$ the service family $f.U \oplus g.W$ is obtained. If a basic action $h.m$ is
applied to a service family $H= \bigoplus_{l=1}^n f_l.U_l$ then two cases are distinguished:  
(i) $h$ equals one of the $f_l$ in which case the method $m$ is applied to $U_l$, 
so that either if $m \in I_m(U)$ a reply is obtained and the state of $U_l$ is updated, 
or otherwise an error occurs, or (ii) none of the $f_l$ equals $h$ in which case an error occurs. 

When computing the application of an instruction sequence to a service family 
(i.e. computing $X \bullet H$) an error is represented by having the empty service family as the result:  
$X \bullet H = \emptyset$. Evaluation of $h.m$ over 
$H= \bigoplus_{l=1}^n f_l.U_l(s_i)$  works fine as long as there is at most a single 
$l$,  for which $h = f_l$ and moreover for that $l$, $m \in I_M(U)$. In that case performing basic action $h.m$ yields reply 
$y^l_m(s_l)$ while changing the state $s_l$ of 
$U_l$ to $e^l_m(s_l)$, leaving the states of the other services in the service family  unmodified. 
In other cases the empty service family is produced.

 In the case of single bit services the inactive service kernel is denoted with $br(\star)$ rather than with $0$, 
i.e. a register containing an error value, so that:
$g.br(b) \oplus g.br(c) = g.br(\star)$. Applying any method to a register containing $\star$ is considered a run time error,
the handling of which depends on the context in principle. In the setting of this paper an error leads to the 
production of the empty service family.

\subsection{Service family restriction}
Let $V$ be a set of foci and $H$ a service family. Then $\partial_V(H)$, the $V$-restriction of $H$,
 results from $H$ by removing 
(i.e. replacing by $\emptyset$) each service $f.U$ in $H$ with $f \in V$). For the special case that $V = \{f\}$ 
we find that each $H$ can be written in one of two forms: $H = \partial_{\{f\}}(K)$ or $H = \partial_{\{f\}}(K) \oplus f.U$.
Service family restriction satisfies some useful equations: 
$\partial_V(H \oplus K) = \partial_V(H) \oplus \partial_V( K)$, 
$\partial_{V \cup W}(H ) = \partial_V\circ \partial_W(H)$, 
$ f \in V \to \partial_V(f.U) = \emptyset$, $f \notin V \to \partial_V(f.U) = f.U$.

\subsection{Basic action interfaces} A basic action (name) is a pair $f.m$ with $f$ a focus and $m$ a method (name). 
A basic action interface is a finite collection of pairs $f.V$ where $f$ is a focus and $V$ is 
a method interface. The notation is simplified by writing $f.V$ for the basic action interface $\{f.V\}$, and by writing $-+-$ for union. Both instruction sequences and service families come with a basic action interface. We write $I(H)$ for the basic action interface of a service family $H$ and $I(X)$ for the basic action interface of an instruction sequence $X$.

For instruction sequences the  interface $I(X)$ collects all focus method pairs that occur in instructions in $X$. 
$I(X)$ is a required interface, as it collects requests (method calls) which an environment is supposed to respond to. Defining equations for $I(-)$ are:
$I(!) = I(\#k) = \emptyset$, $ I(g.m) = I(+g.m) =I(-g.m) = g.\{m\}$, and $I(X;Y) = I(X) + I(Y)$, in combination with
$g.V + g.W = g.(V\cup W)$. 

For a service family $H$, $I(H)$ collects the method calls to which $H$ is able to respond. 
For service families the interface definition is less straightforward than for instruction sequences: 
$I(\emptyset) = \emptyset$, $I(g.U) = g.I_m(U)$, 
$I(\partial_{\{g\}}(H) \oplus g.U) = I(\partial_{\{g\}}(H)) + I(g.U)$. From these equations it follows that 
$I(g.U \oplus g.V) = g.\emptyset$, and therefore distribution of $I$ over $U$, i.e. $I(g.U \oplus g.V) = I(g.U) + I(g.V)$, 
fails if $I_m(U) \neq \emptyset$ 
and if $I_m(V) \neq \emptyset$.
\section{Interfaces as constraints on instruction sequences}
\label{Interfaces}

Given a basic action interface $I$  the collection of PGA instruction sequences for acting on single bit registers
$X$ such that $I(X) \subseteq I$
 is denoted $\mathrm{IS}_{br}(I)$. Membership of $\mathrm{IS}_{br}(I)$ for an appropriate basic 
 action interface $I$ is a useful constraint on an instruction sequence. 
We will provide several examples of such constraints in the following Paragraphs of this Section.

Interfaces are partially ordered by inclusion ($I \subseteq J$, $I$ is a subinterface of $J$, 
$I$ is contained in $J$, $J$ includes $I$).
An interface $I$ may serve as a constraint on instruction sequences, in particular the 
requirement that the required interface of the instruction sequence $X$  is not too large: $I(X) \subseteq I$. 

At the same time a basic action interface $I$ may serve as a constraint on a service family $H$ 
on which
$X$ is supposed to operate: $I \subseteq I(H)$, that is the requirement that the 
provided interface of $H$ is not too small.
We will provide four examples of the use of interfaces  in the following Paragraphs.
\subsection{Alternative initialisation of output registers}
An obvious extension of the instruction set outlined in Paragraph~\ref{objectives} above is to allow to make use of registers $out1{:}k$ which have $1$ as the  initial content. Allowing 1-initialised output registers extends the class of instruction 
sequences in such a manner that there is a gain of expressiveness.

To see this improvement consider the function $F(-)$ with a single input $in{:}1$ which 
takes constant value $0$.  Working with interface $J = in{:}1.i/i + out0{:}1/1$  the mere termination instruction $!$
constitutes an instruction sequence that computes $F$ with LLOC 1. Alternatively if an instruction sequences is sought for
 in $J^\prime = in{:}1.i/i + out1{:}0/0$, a longer instruction sequence (LLOC 2) such as  $out1{:}1.0/0;!$, is required.

\subsection{Bit complementation}
We will consider  the  function $F\colon \{0,1\} \to \{0,1\}$ given by $F(x) = 1-x$. $F$ 
represents complementation (negation). 

Below seven  instruction sequences each of which compute $F$ are listed. 
By imposing restrictions on the basic action interface serving as a constraint the differences between these options 
for implementing complementation of a single bit can be made explicit.  

The role $out$ stands for a register which serves as an output as it won't be read, but which may have initial value 0 or 1.
Thus a single bit register with focus say $out{:}1$ may have arbitrary initialisation.
\begin{itemize}
\item $I_1 = inout{:}1.M_{16}$. Both inputs and outputs reside in $inout{:}1$. The instruction sequence 
$X_1 = inout{:}1.1/c;!$ computes $F$
and is a shortest possible program because at least one basic action 
needs to be applied to the input and a termination instruction must be included.
\item $I_2 = in{:}1.M_{16} + out0{:}1.M_{16}$. In this case output is placed in a different register serving as  
an output register only. $X_2 = +in{:}1.i/i;out0{:}1.1/1;!$ computes $F$. 
Moreover a shorter implementation cannot be found: the input needs to be read and some 
writing of outputs is unavoidable as well as a termination instruction.

\item $I_3 = in{:}1.M_{16} + out{:}1.M_{16}$. 

The instruction sequence 
  $X_3 = +in{:}1.i/i;+out{:}1.0/1;out{:}1.0/0;!$ computes $F$. A shorter instruction sequence for computing $F$ when
implemented under the constraints of basic action interface $I_3$ does not exist. An input instruction is necessary, 
and both values must be written by some output instruction because both outputs can arise while the initial 
content of the output register is not known in advance.

\item $I_4 = inout{:}1.\{i/i,1/1,1/0\}$.  Now the instruction sequence \\
$X_4 = +inout{:}1.i/i;-inout{:}1.1/1;inout{:}1.1/0;!$ is in $\mathrm{IS}_{br}(I_4)$ and computes $F$
and it is easy to see that it constitutes  a shortest possible program in $\mathrm{IS}_{br}(I_4)$ for that task.
\item $I_5 = in{:}1{:}.\{c/0\} + out0{:}1.\{1/c\}$. Now  
$X_5= +in{:}1.c/0;out0{:}1.0/c;!$ is in $\mathrm{IS}_{br}(I_5)$, computes $F$, and as such has minimal LLOC for that task. 

\item $I_6 = in{:}1.\{i/i\} + out{:}1.\{i/0,i/1\}$.   A shortest implementation of $F$ under these constraints is:
$X_6=out{:}1.i/0;-in{:}1.i/i;out{:}1.i/1;!$. 

\item $I_7 = in{:}1.\{i/i\} + out{:}1.\{i/c\}$.  $F$ is computed by \\$X_7 = -out{:}1.i/c;out{:}1.i/c;-in{:}1.i/i;out{:}1.i/c;!$.
\end{itemize}

\begin{proposition}
\label{complementation}
As a single pass instruction sequence computing $F$ under the 
constraint that $I(X) \subseteq in{:}1.\{i/i\} + out{:}1.\{i/c\}$ $X_7$ minimises LLOC.
\end{proposition}

\begin{proof}
$\LLOC(X_7)=5$ and therefore consider an implementation $Y$ of $F$ which has 4 instructions, say $Y = u_1;u_2;u_3;u_4$. 
We may assume that $u_4 = !$ because otherwise $u_4$ cannot be performed unless a faulty termination takes 
place with the effect that $Y$may be simplified to three or even fewer instructions while still computing $F$. 
That is impossible because at least one read instruction on $in{:}1$ and two different write instructions 
on $out{:}1$ (for $0$ and for $1$) must appear in $Y$. 
This observation also implies that the $\LLOC(Y)$ is at least 4. So $Y = u_1;u_2;u_3;!$. 
If $u_3$ were an input instruction, the output of $Y$ is independent of the input, which is not the case. 

Thus $u_3 \in \{ out1{:}1. i/c, +out1{:}1.i/c, -out{:}.i/c\}$.  
Now a case distinction on $u_1$ reveals that
$u_1 = -out1{:}1.i/c$ fails because  starting with $out{:}1 = 1$ the second instruction  
is skipped and no input action is performed. Similarly $u_1 = +out1{:}1.i/c$ fails because  starting with $out{:}1 = 0$ 
the second instruction  is skipped and no input action will be performed. 
If $u_1 = out1{:}1.i/c$ then the collection of results for $out1{:}1 = 0$ and for $out1{:}1 = 1$ is 
left unchanged when deleting $u_1$ so that a shorter instruction sequence $u_2;u_3;!$ 
also implements $F$ which has been ruled out already. 

Thus $u_1$ is an input instruction. It must be a test instruction because otherwise the output will not depend on the input. Let $u_1 = +in{:}1.i/i$, the symmetric case $u_1 = -in{:}1.i/i$ can be dealt with similarly. Upon input $in{:}1=0$
the computation of proceeds with $u_3;!$ with $u_3\in \{ out1{:}1. i/c, +out1{:}1.i/c, -out{:}.i/c\}$. Consider the case with initial value $1$ for $out{:}1$ then for each option for $u_3$ the resulting value for $out{:}1$ is 
$out{:}1 =0$ instead of the required output $out{:}1 = 1$. In all cases a contradiction has been derived thus contradicting the initial assumption that $Y$ with LLOC equal to 4 computes $F$.
\end{proof}

The following fact admits an easy but tedious proof, the details of which are left aside.
\begin{proposition} For each basic action interface $I \subseteq in{:}1.M_{16} + out{:}1.M_{16}$ the following holds:  
if $F$ (complementation) can computed by a finite single pass instruction sequence in $\mathrm{IS}_{ba}(I)$ then
$F$ can be computed by an instruction sequence with  LLOC at most 5.
\end{proposition}

\subsection{Parity checking}
The second  example of the use of interfaces as constraints concerns 
the role of auxiliary registers in single pass  instruction sequences for computing multivariate functions on Booleans.
We will survey the results of~\cite{BergstraM2016} while reformulating these in terms of interfaces.

Let $I^n = \sum_{l =1}^n in{:}l.\{i/i\} + out0{:}1.\{1/1\}$. The function $P$ on bit sequences is given by: $P(0) = 0, P(1) = 1,  
P(0,\alpha) = P(\alpha), P(1,\alpha) = 1 - P(\alpha)$. $P(-)$ determines the parity of a sequence of bits.
We are interested in instruction sequences for computing $P(-)$ from inputs
stored in input registers with focus $in{:}1,\ldots, in{:}n$. 

From~\cite{BergstraM2016} we take that the instruction sequence $\mathrm{PARIS}^0_n$ 
computes parity for $n$ bits: 
\[\mathrm{PARIS}^0_0 =\, !\]
\[\mathrm{PARIS}^0_1 = +in{:}1.i/i;out0{:}1.1/1;!\] and for $n >1$:
 $\mathrm{PARIS}^0_n = +in{:}1.i/i:\gsc_{l=2}^n (\#4;+in{:}l.i/i;\#3;\#3;-in{:}l.i/i);out0{:}1.1/1;!$
 
Formalisation of the fact that these instruction sequences perform parity checking looks as follows in the 
notation of~\cite{BergstraM2012b} and~\cite{BergstraM2018}.

For all $n \geq 1$ and for all bit sequences $(b_1,\ldots b_n)$:
\[\partial_{\{in{:}1,\dots,in{:}n\}}(|\mathrm{PARIS}^0_n|\bullet 
	(out0{:}1.br(0) \oplus \bigoplus_{k=1}^n in{:}k.br(b_n) ) = out{:}1.br(P(b_1,\ldots b_n))\]
For $n > 0$ we find that $\LLOC(\mathrm{PARIS}^0_n) = 5 (n-1) + 3 = 5.n -2$. 

Next consider
the interface $I^n_a = I^n+aux{:}1.\{i/c,i/i\}$ and the instruction sequences $\mathrm{PARIS}^1_n$ with
$I(\mathrm{PARIS}^1_n) \subseteq I_a$. 
\[\mathrm{PARIS}^1_0 = !\]
\[\mathrm{PARIS}^1_1 = \mathrm{PARIS}^0_1\]
 and for $n > 1$:
$\mathrm{PARIS}^1_n = \gsc_{l=1}^n (+in{:}l.i/i;aux0{:}1.i/c);+aux0{:}1.i/i;out0{:}1.1/1;!$.

In~\cite{BergstraM2016} it is shown that $\mathrm{PARIS}^1_n$ computes $P(-)$ on $n$ inputs. In
formal notation this reads:
\[\partial_{\{in{:}1,\dots,in{:}n, aux0{:}1\}}(|\mathrm{PARIS}^1_n|\bullet 
	(out0{:}1.br(0)\oplus aux0{:}1.br(0) \oplus \bigoplus_{k=1}^n in{:}k.br(b_n) ) =\]\[ out{:}1.br(P(b_1,\ldots b_n))\]
For $n>1$ $\LLOC(\mathrm{PARIS}^1_n) = 2n+3$. 
Moreover it was shown in~\cite{BergstraM2016} that from $n=7$ upwards, 
each instruction sequence in $\mathrm{IS}_{ba}(I^n_a)$
which computes $P(-)$ has more instructions than $2n+3$, thereby establishing that the availability of
an initialised auxiliary register in this particular case 
enables to write a shorter instruction sequence.
We are unaware, however, of conclusive answers to the following questions.
\begin{problem} What is the lowest LLOC size of a single pass instruction sequence in 
$\mathrm{IS}_{br}(I^n)$which computes $P(-)$ on $n$ single bit registers? 
\end{problem}
\begin{problem}
What is the lowest LLOC size of a single pass instruction sequence in 
$\mathrm{IS}_{br}(I^n_a)$which computes $P(-)$ on $n$ single bit registers?
\end{problem}

\subsection{Proving the expressive power of large jumps}
\label{simplifying}
We return to Example~\ref{Exam2} and prove a partial result the proof of which becomes manageable 
by imposing significant constraints on the required interface of instruction sequences involved. 
The interface constraint excludes the use of auxiliary registers and imposes that input actions 
complement the input at the same time.
We recall from Example~\ref{Exam2} the instruction sequences 
\noindent
$X_E^k =+inout{:}1.i/c;\#2k+2;_{l=1}^{k} (+inout\_a{:}l.i/c;out0\_a{:}l.1/1);!;$\\
$~\quad\quad \quad  ;_{l=1}^{k} (+inout\_b{:}l.i/c;out0\_b{:}l.1/1);!$.

\begin{proposition} 
\label{largejumps}
\label{Exam2ctd}
Let $k > 3$ and suppose that $X$ is a finite single pass instruction sequence such that the following four conditions are met:
\begin{enumerate}
\item $I(X) \subseteq I(X_E^k)$,
%$X \in I_{ba}( 
%inout{:}1.\{i/c\} \cup \bigcup_{l=1}^k(inout\_a{:}l.\{i/c\} \cup 
%	inout\_b{:}l.\{i/c\} \cup out0\_a{:}l.\{1/1\} \cup out0\_b{:}l.\{1/1\})$
\item $X$ computes the same function as $X_E^{k}$,
\item $\LLOC(X) \leq \LLOC(X_E^k)$,
and 
\item the only termination instruction in $X$ is its final instruction,
\end{enumerate}
then $X$  contains at least one jump instruction of length  $k/2$ or more.
\end{proposition}

\begin{proof} $X$ For $X$ to compute the same transformation as $X^k_E$ it must contain
at least $4k+1$ instructions with precisely that number of different foci as occur in the basic instructions of $X^k_E$. 
$X$ may in addition contain at most two more instructions each of which are either a basic instructions or  a jump, 
or a termination instruction.

The number of basic instructions that must be performed in order to compute the required 
transformation ranges between a minimum of $k+1$, reading $in{:}1$ and either reading all inputs $inout\_a{:}l$ 
or reading all inputs $inout\_b{:}l$ while not writing any outputs, which comes about upon taking all inputs equal to $0$, 
and  a maximum of $2k+1$, involving the maximum most $k$ updates of output registers, 
a maximum which is reached only when all inputs that are read have value 1. 
The total number of instructions that are performed 
during a run is between $k+2$ (the minimal number of basic instructions and a termination instruction) and $2k+4$ 
(the maximally required number of basic instructions plus termination plus two more instructions as mentioned).

If $X$ only contains a termination at the final position and at most two  jumps  $\#l_1$ and $\#l_2$ with 
$l_1 + l_2 \leq k$, or a single jump $\#l$ with $l \leq k$,  then at least 
$(4k+2 -l)/2 \geq \sfrac{3}{2} \cdot k+1> k+2$  instructions are performed during each run of $X$, 
as each basic instruction can at best have its successor skipped, and the  jump instruction(s) 
skips at most $k-1$ of its (both) successors, and importantly no intermediate termination can take place. Now consider
 a run where no non-zero inputs  have been observed and hence no output was written. 
It follows that because this run is performing more basic instructions than $k+1$, and because it cannot perform any output actions, as setting an output to 1 cannot be undone, said run must perform a read from the same input 
more than once, and in fact that input must be read three times at least because its content needs to be set to 1 at the end. 

We find that for some focus with role $in$ or $inout$ three basic instructions are present in $X$ so that 
no room is left for the jump instruction(s) which we assumed to be present. We find a contradiction which concludes the proof.
\end{proof}

\section{Addition of natural numbers}
Addition of natural numbers in binary notation and with equal numbers of bits takes $2 n$ inputs and 
produces $n+1$ outputs. I will present some instruction sequences for addition relative to different interfaces.

Let  $I_n = \sum_{l =1}^n (in\_a{:}l.\{i/i\} +in\_b{:}l.\{i/i\})$. First we consider the interface 
$I_n^\prime = I_n +  \sum_{l =1}^n (out0{:}l.\{1/1\}) + out0{:}(n+1).\{1/1\} + aux0{:}1.\{i/0,1/1\}$.
The following single pass instruction sequences $Add_n$ for this task, relative to $I_n^\prime$ 
achieve $\LLOC(Add_n) =14  n + 3$. The auxiliary
register $aux{:}1$ is used as a carry.\\

\noindent $Add_n=$
%$ (= \gsc_{k=1}^{n}(P_k);P_{n+1}) = $\\
$\gsc_{k=1}^{n}( $\\
$~\quad +in\_a{:}k.i/i;-in\_b{:}k.i/i;\#4;+aux0{:}1.i/1;\#9;\#9;$\\
$~\quad -in\_a{:}k.i/i;+in\_b{:}k.i/i;\#4;+aux0{:}1.i/0;\#3;\#3;$\\
$~\quad -aux0{:}1.i/i;out0{:}k.1/1;$\\
$~\quad );$
$+aux0{:}1.i/0;out0{:}(n+1).1/1;!$\\
                         
\noindent Working over $I_n^{\prime\prime}=I^n +  \sum_{l =1}^n (out0{:}l.\{1/1\}) + out0{:}(n+1).\{i/i,i/1,i/0\}$ 
the carry can be identified with 
$out0{:}(n+1)$, thereby reducing the number of instructions to $14 n +1$.\\

\noindent $Add^\prime_n=$
%$ (= \gsc_{k=1}^{n}(P_k);P_{n+1}) = $\\
$ \gsc_{k=1}^{n}( $\\
$~\quad +in\_a{:}k.i/i;-in\_b{:}k.i/i;\#4;+out0{:}(n+1).i/1;\#9;\#9;$\\
$~\quad -in\_a{:}k.i/i;+in\_b{:}k.i/i;\#4;+out0{:}(n+1).i/0;\#3;\#3;$\\
$~\quad -out0{:}(n+1).i/i;out0{:}k.1/1;$\\
$~\quad );!$\\

\noindent Making use of the fact that for $k=1$ the check on $out0{:}n+1$ is 
redundant because no assignment to it has yet been made
we find, again working relative to $I_n^{\prime\prime}$, while assuming $n> 0$, that $Add^{\prime\prime}_n$ 
with LLOC $ 14n -5$: \\   

\noindent $Add^{\prime\prime}_n=$
%$ (= \gsc_{k=1}^{n}(P_k);P_{n+1}) = $\\
$+in\_a{:}1.i/i;-in\_b{:}1.i/i;\#3;out0{:}(n+1).i/1;\#4;-in\_a{:}1.i/i;+in\_b{:}1.i/i;out0{:}1.1/1;$\\
$~\quad \gsc_{k=2}^{n}( $\\
$~\quad +in\_a{:}k.i/i;-in\_b{:}k.i/i;\#4;+out0{:}(n+1).i/1;\#9;\#9;$\\
$~\quad -in\_a{:}k.i/i;+in\_b{:}k.i/i;\#4;+out0{:}(n+1).i/0;\#3;\#3;$\\
$~\quad -out0{:}(n+1).i/i;out0{:}k.1/1;$\\
$~\quad );!$  \\        

\noindent An remaining question is this:
 \begin{problem} Are there for any $n>0$ single pass instruction sequences $Add^\star_n$ for addition 
 over interface $I^{\prime\prime}_n$ with $\LLOC(Add^\star_n) <14n -5$? If so, what are the 
 shortest single pass instruction sequences for addition for this interface?
 \end{problem}
\subsection{Allowing bit complementation}
Now let $I_n^{\prime\prime\prime} = I_n + \sum_{l =1}^{n+1} out0{:}l.\{i/c\} + out0{:}(n+1).\{i/0\}$. 
W.r.t. $I_n^{\prime\prime\prime} $ 
the following instruction sequences $Add^{\prime\prime\prime}$ with $\LLOC(Add^{\prime\prime\prime})= 8n $ 
implement addition. Here we do without a carry bit.\\
 
\noindent $Add^{\prime\prime\prime}_n=$
%$ (= \gsc_{k=1}^{n}(P_k);P_{n+1}) = $\\
$+in\_a{:}1.i/i;out0{:}1/c;-in\_b{:}1.i/i;\#3;+out0{:}i/c;out0{:}2.i/c;$\\
$~\quad \gsc_{k=2}^{n}( $\\
%$~\quad +out0{:}(n+1).i/0;out0{:}i/c;$\\
$~\quad -in\_a{:}k.i/i;\#3;+out0{:}i/c;out0{:}(n+1).i/c;$\\
$~\quad -in\_b{:}k.i/i;\#3;+out0{:}i/c;out0{:}(n+1).i/c;$\\
$~\quad );!$  \\   

\begin{problem}
\label{addition}
For each $n>0$: is there a single pass PGA instruction $X$ implementing  
addition of two $n$ bit naturals over the interface
 $I_n^{\prime\prime\prime}$ with $\LLOC(X)< 8n$? And more generally what is the 
 lowest LLOC that can be achieved for this task?
 \end{problem}
 
 \subsection{concluding remarks}
 Finding shortest possible instruction sequences for addition is a prerequisite for similar work on multiplication which  is
 a viable topic of future research. LLOC in combination with the notion of quantitative expressiveness presents 
 a promising approach to complexity theory for instruction sequences and allows 
 for investigation which is not primarily focused on asymptotics but is rather more of a combinatorial nature. 
 The use of interfaces
 allows essential flexibility concerning matters of quantitative expressiveness.


\addcontentsline{toc}{section}{References}
\begin{thebibliography}{58}
 \bibitem{BergstraL2002}
J.A. Bergstra and M.E. Loots. 
\newblock Program algebra for sequential code.
\newblock \emph{Journal of Logic and Algebraic Programming} 50 (2), 125--156, (2002).

%\bibitem{BergstraB2011}
%J.A. Bergstra and I. Bethke. 
%\newblock Straight-line instruction sequence completeness for total calculation on meadows. 
%\newblock \emph{Theory of Computing Systems} 48, 840--864, (2011).

\bibitem{BergstraM2007}
J.A. Bergstra and C.A. Middelburg. 
\newblock Instruction sequences with indirect jumps.
\newblock \emph{Scientific Annals of Computer Science} 17, 19--46, (2007),
 (also \url{http://arxiv.org/abs/0711.0829v2}).



\bibitem{BergstraM2012}
J.A. Bergstra and C.A. Middelburg. 
\newblock On the expressiveness of single pass instruction sequences. 
\newblock \emph{Theory of Computing Systems} 50 (2), 313--328, (2012), (also \url{http://arxiv.org/abs/0810.1106v3}).

\bibitem{BergstraM2012b}
J.A. Bergstra and C.A. Middelburg. 
\newblock Instruction sequence processing operators
\newblock \emph{Acta Informatica}, vol. 49, no. 3, pp. 139--172, (2012), (also \url{http://arxiv.org/abs/0910.5564v5}).

\bibitem{BergstraM2012c}
J.A. Bergstra and C.A. Middelburg.                                                                                                                                                   
\newblock Indirect jumps improve instruction sequence performance.
\newblock \emph{Scientific Annals of Computer Science} 22 (2), 253--265, (2012),
(also \url{http://arxiv.org/abs/09090.2089v3})


\bibitem{BergstraM2012d}
J.A. Bergstra and C.A. Middelburg.                                                                                                                                                   
\newblock Instruction Sequences for Computer Science.
\newblock Atlantis Publishing, ISBN:9491216643 (2012).

\bibitem{BergstraM2014a}
J.A. Bergstra and C.A. Middelburg. 
\newblock Instruction sequence based non-uniform complexity classes.
\newblock \emph{Scientific Annals of Computer Science} 24 (1), 47--89, (2014),
(also \url{http://arxiv.org/abs/1301.3297.v2}).

\bibitem{BergstraM2016}
J.A. Bergstra and C.A. Middelburg. 
\newblock Instruction sequence complexity of parity. 
\newblock \emph{Fundamenta Informaticae} 149 (3), 297--309, (2016), 
(also \url{http://arxiv.org/abs/1412.6787}).

\bibitem{BergstraM2016b}
J.A. Bergstra and C.A. Middelburg. 
\newblock Instruction sets for Boolean registers in program algebra.
\newblock \emph{Scientific Annals of Computer Science} 26 (1), 1--26, (2016),
(also \url{http://arxiv.org/abs/1502.00238.v2}).


\bibitem{BergstraM2018a}
J.A. Bergstra and C.A. Middelburg. 
\newblock Instruction sequences expressing multiplication algorithms.
\newblock \emph{Scientific Annals of Computer Science} 28 (1), 39--66, (2018).
 (also:
\url{http://arxiv.org/abs/1312.1529v4}).


\bibitem{BergstraM2018}
J.A. Bergstra and C.A. Middelburg. 
\newblock A short introduction to program algebra with instructions for Boolean registers.
\emph{Computer Science Journal of Moldovia} vol. 26 no. 3 (78), 2018 (also:
\url{http://arxiv.org/abs/1808.04264.v2})

%\bibitem{BergstraP2007}
%J.A. Bergstra and A. Ponse. 
%\newblock Projection semantics for rigid loops.
%\url{http://arxiv.org/abs/0707.1059}  (2007).
%

\bibitem{BergstraP2007b}
J.A. Bergstra and A. Ponse. 
\newblock Execution architectures for program algebra.
\emph{Journal of Applied Logic} vol. 5 no. 1, 175--192, (2007).

%\bibitem{BergstraP2013}
%J.A. Bergstra and A. Ponse. 
%\newblock Periodic single-pass instruction sequences.
%\url{http://arxiv.org/abs/0819.1151.v2}  (2013).

\bibitem{Chaitin2003}
G. Chaitin. 
\newblock From philosophy to program size.
\emph{Lect. Notes on algorithmic information theory} arXiv:math/003352v2. (2003)

\bibitem{Constable1971}
R.L. Constable. 
\newblock Subrecursive programming languages II, On program size.
\emph{J. of Computer and System Sciences} 5, 315--334, (1971).




\bibitem{KatseffS1981}
H.P. Katseff and M. Sipser. 
\newblock Several results in program size complexity.
\emph{Theoretical Computer Science} vol. 15, 291--309, (1981).


\bibitem{Massalin1987}
H.\ Massalin. 
\newblock Superoptimizer: a look at the smallest program.
\emph{Proc. 2th ASPLOS} 122--127, (1987).

\bibitem{Mehlhorn1982}
K. Mehlhorn. 
\newblock On the program size of perfect and universal hash functions.
\emph{Proc. 23th FOCS, ACM} 170--175, (1982).


\bibitem{Meyer1972}
A.R. Meyer. 
\newblock Program size in restricted programming languages.
\emph{Information and Control} vol. 21 no. 4, 382--394, (1972).

\bibitem{NguyenDTB2007}
V.\ Ngyyen, S.\ Deeds-Rubin, T.\ Tan, and B.\ Boehm.
\newblock A SLOC counting standard. Center for Systems and Software Engineering, 
University of Southern California,
(2007).

\end{thebibliography}
\end{document}